\documentclass[11pt]{article}
\usepackage{color,soul}
\usepackage{titlesec}
\usepackage{hyperref}
\usepackage{dirtytalk}
\usepackage{lipsum}
\usepackage[thinlines]{easytable}
\usepackage{ntheorem}
\usepackage{mathtools}
\usepackage{graphicx}
\usepackage{amssymb}
\usepackage{mathrsfs}
\usepackage{calc}

\newlength{\depthofsumsign}
\usepackage{algorithm2e}
\RestyleAlgo{ruled}
\SetKwInput{KwData}{The data}
\SetKwInput{KwResult}{The result}
\usepackage{graphicx}
\usepackage{amsmath}
\usepackage{amssymb}
\usepackage{amsfonts}
\usepackage{array}
\usepackage{longtable}
\usepackage{tikz}
\usepackage{hhline}

\usepackage{multirow}
\usepackage{subfig}
\usepackage{hyperref}
\usepackage{mathtools}
\usepackage{tikz}
\usepackage{makecell}
\usepackage{pgfplots}
\usepackage{url}
\usepackage{afterpage}
\usepackage{fancyhdr}
\usepackage{setspace}
\usepackage{bm}

\usepackage{diagbox}

\newtheorem{proposition} {Proposition}
\newtheorem*{proposition-non}{Proposition}

\newtheorem{theorem}{Theorem}

\newtheorem{conjecture}{Conjecture}
\newtheorem{remark}{Remark}
\newtheorem{example}{Example}

\DeclareMathOperator*{\argmax}{arg\,max}

\allowdisplaybreaks

\newenvironment{proof}{\noindent{\bf Proof:}\indent}%
                      {\hfill $\Box$\par}

\newcommand{\sym}[1]{{\sf #1}}
\title{A Lower Bound on the Constant in the Fourier Min-Entropy/Influence Conjecture}
\author{Aniruddha Biswas and Palash Sarkar \\
Indian Statistical Institute \\
203, B.T.Road, Kolkata \\
India 700108. \\
Email: \{anib\_r, palash\}@isical.ac.in
}
\date{\today}

\begin{document}

\maketitle

\begin{abstract}
	We describe a new construction of Boolean functions. A specific instance of our construction provides a 30-variable Boolean function having
	min-entropy/influence ratio to be $128/45 \approx 2.8444$ which is presently the highest known value of this ratio that is achieved by any Boolean function. 
	Correspondingly, $128/45$ is also presently the best known lower bound on the universal constant of the Fourier min-entropy/influence conjecture. 
\end{abstract}

\section{Introduction \label{sec-intro} }

A longstanding open problem in the field of analysis of Boolean functions is the Fourier Entropy/Influence (FEI) conjecture made by Friedgut and Kalai in 
1996~\cite{friedgut1996every}. The FEI conjecture states that there is a universal constant $C$ such that $H(f)\leq C\cdot\sym{Inf}(f)$ for any Boolean function $f$, 
where $H(f)$ and $\sym{Inf}(f)$ denote the Fourier entropy and the total influence of $f$ respectively. For a nice discussion on the importance and applications of the
FEI conjecture, we refer to the blog post by Kalai~\cite{Kalai07}.
The conjecture was verified for various families of Boolean functions (e.g., symmetric functions~\cite{o2011fourier}, read-once formulas~\cite{o2013composition,chakraborty2016upper}, 
decision trees of constant average depth~\cite{wan2014decision}, read-$k$ decision trees for constant $k$~\cite{wan2014decision}, 
functions with exponentially small influence or with linear entropy~\cite{shalev2018fourier}, random linear threshold functions~\cite{DBLP:conf/latin/0001KKLS18},
cryptographic Boolean functions~\cite{gangopadhyay2014fourier}, random functions~\cite{das2011entropy}),
but is still open for the class of all Boolean functions. See~\cite{kelman_gafa2020} for the most recent work towards settling the conjecture.

There has also been research in obtaining lower bounds on the constant $C$ in the FEI conjecture. To show that $C$ is at least some value $\delta$ it is sufficient
to show the existence of a Boolean function whose entropy/influence ratio is $\delta$. 
The first lower bound of $4.615$ was obtained by O'Donnell et al. in~\cite{o2011fourier}. Later O'Donnell and Tan~\cite{o2013composition} provided a recursive construction 
of Boolean functions which showed how to construct a function for which the value of the entropy/influence ratio is at least $6.278944$~\cite{hod2017improved}. The presently 
best known lower bound on $C$ is $6.454784$. This bound was shown by Hod~\cite{hod2017improved} using an extensive asymptotic analysis.

The Fourier min-entropy/influence (FMEI) conjecture was put forward by O'Donnell et al. in 2011~\cite{o2011fourier}. The FMEI conjecture states that there is a universal constant
$D$ such that $H_{\infty}(f)\leq D\cdot\sym{Inf}(f)$ for any Boolean function $f$, where $H_{\infty}(f)$ is the Fourier min-entropy of $f$.
Since $H_{\infty}(f)\leq H(f)$,
the FMEI conjecture is weaker than the FEI conjecture in the sense that settling the FEI conjecture will also settle the FMEI conjecture, but the converse is not
true. It was observed in~\cite{DBLP:conf/latin/0001KKLS18,o2011fourier} that as a consequence of the Kahn-Kalai-Linial theorem~\cite{DBLP:conf/focs/KahnKL88} 
the FMEI conjecture holds for monotone functions and linear threshold functions. The FMEI conjecture for ``regular'' read\mbox{-}$k$ DNFs 
%(where regular means each term in the DNF has more or less the same number of variables, see~\cite{shalev2018fourier} for a precise definition) 
was established by Shalev~\cite{shalev2018fourier}. More recently, Arunachalam et al.~\cite{DBLP:journals/toct/ArunachalamCKSW21} showed that the FMEI conjecture
holds for read\mbox{-}$k$ DNF for constant $k$.

To the best of our knowledge, till date there has been no work on obtaining lower bounds on the universal constant of the FMEI conjecture. Since the FMEI conjecture 
is weaker than the FEI conjecture, any upper bound on the universal constant of the FEI conjecture is also an upper bound on the universal constant of the FMEI conjecture. 
This, however, does not hold for lower bounds, i.e. a lower bound on the universal constant of the FEI conjecture is not necessarily a lower bound on the universal conjecture
of the FMEI conjecture. 

The purpose of the present paper is to obtain a lower bound on the universal constant $D$ of the FMEI conjecture. As in the case of the FEI conjecture, to show that $D$ 
is at least $\delta$, it is sufficient to show the existence of a Boolean function for which the min-entropy/influence ratio is $\delta$. An exhaustive search over
all $n$-variable Boolean functions, with $1\leq n\leq 5$, shows that the maximum value of min-entropy/influence ratio that is achieved by functions of at most 5 variables
is $16/7\approx 2.285714$.
Since an exhaustive search becomes infeasible for $n\geq 6$, it is required to obtain some method of constructing Boolean functions for which the min-entropy/influence
ratio is greater than 16/7. 

For the above mentioned purpose, we first considered the recursive construction of O'Donnell and Tan~\cite{o2013composition}, since this construction proved to be useful
for showing a lower bound on the constant of the FEI conjecture. To analyse this construction in the context of the FMEI conjecture, we derived an expression for 
the min-entropy of the functions obtained using this construction. Since the construction is recursive, one needs an initial function to start the recursion. We 
performed an exhaustive search
over all possible 5-variable initial functions. This yielded a 25-variable function having min-entropy/influence ratio equal to $512/225\approx 2.275556$.
This unfortunately is not useful since 512/225 is less than 16/7, the maximum value of min-entropy/influence ratio that is obtained by exhaustive search over all
5-variable functions. The 25-variable function is obtained in the first step of the O'Donnell-Tan recursion. Considering further steps of the recursion does not
result in a higher value of the min-entropy/influence ratio. 
%We identified an alternative recursive construction of Boolean functions which provides 
%a lower bound on the constant of the FEI conjecture which is equal to that obtained from the O'Donnell-Tan construction. This alternative construction, however, does 
%not improve upon the lower bound on the constant of the FMEI conjecture that is obtained from the O'Donnell-Tan construction. 
Further, we did not find any way to 
apply the asymptotic constructions given by Hod~\cite{hod2017improved} in the context of the FEI conjecture for obtaining lower bounds on the constant in the FMEI conjecture.

Our main result is a new construction of Boolean functions. In simple terms, the construction takes an $n$-variable function $g$ and constructs an
$(n+1)$-variable palindromic function $g_0$. An $n(n+1)$-variable function $G_0$ is then constructed by taking the disjoint composition of $g_0$ and $g$ (see~\eqref{eqn-disj-comp}
in Section~\ref{sec-prelim} for the definition of disjoint composition). Under 
certain conditions on $g$, the min-entropy/influence ratio of $G_0$ is greater than that of $g$. By searching over all appropriate 5-variable functions
$g$, we obtain a 30-variable function $G_0$ having min-entropy/influence ratio to be equal to $128/45\approx 2.844444$. In fact, we obtain a total of 384 such 
functions $G_0$. The value $128/45$ is presently the highest achieved value of min-entropy/influence ratio and correspondingly is presently the best known lower bound
on $D$. 

At this point of time, there is no clear evidence of whether the FMEI conjecture is indeed true or not. Our work provides the first step towards an
understanding of this difficult conjecture. Assuming that the conjecture holds, it is of intrinsic mathematical interest to know the value of the
universal constant in the conjecture. Again, our result provides the first step in this direction.
We note that there has been decades long interest in improving values of constants in mathematical results such as the Berry-Esseen 
theorem\footnote{See \url{https://en.wikipedia.org/wiki/Berry\%E2\%80\%93Esseen_theorem} (accessed on 8th November, 2023) for a nice discussion on how the bounds
on the constant has improved over the years.}.

In the final section, we provide a brief description of some experiments that we have carried out for symmetric and rotation-symmetric Boolean functions. 
Based on these experiments, we put forward a new conjecture on entropy/influence and the min-entropy/influence ratios of symmetric Boolean functions. 
%The FEI (and hence the FMEI) conjecture for symmetric Boolean functions was shown to hold in~\cite{o2011fourier}. Our attempts at using the proof 
%technique from~\cite{o2011fourier} to settle the new conjecture turned out to be unsuccessful. It seems that new ideas are required to settle the conjecture.

In Section~\ref{sec-prelim}, we describe the formal background and the notation. The technique of disjoint composition is required 
for our construction. In Section~\ref{sec-disj-comp-WT}, we derive an expression for the min-entropy of disjoint composition.
The recursive construction of O'Donnell and Tan~\cite{o2013composition} is analysed in Section~\ref{sec-lo-bnd} and shown to
be not useful for the FMEI conjecture. Section~\ref{sec-palin} presents the main construction of the paper and the description of the Boolean functions
achieving the presently best known value of the min-entropy/influence ratio. In Section~\ref{sec-srch-res} we briefly describe our search experiments with
symmetric and rotation-symmetric Boolean functions and state the new conjecture. Finally, in Section~\ref{sec-conclu} we provide concluding remarks.

%\section{Background and notation \label{sec-prelim} }
%\section{Min-Entropy of disjoint composition\label{sec-disj-comp-WT}}
%\section{Constructions \label{sec-lo-bnd}}
%\section{Construction from palindromic functions \label{sec-palin}}
%\section{Some further search results\label{sec-srch-res}}
%\section{Conclusion \label{sec-conclu}}

\section{Background and notation \label{sec-prelim} }
Let $\mathbb{F}_2=\{0,1\}$ denote the finite field consisting of two elements with addition represented by $\oplus$ and multiplication by $\cdot$; often, for 
$x,y\in \mathbb{F}_2$, the product $x\cdot y$ will be written as $xy$. The field of real numbers will be denoted by $\mathbb{R}$ and all logarithms are to the base 2.
%and expression of the form $0\log 0$ and $0\log \frac{1}{0}$ are to be interpreted as $0$.

For a positive integer $n$, by $[n]$ we will denote the set $\{1,\ldots,n\}$. For $\mathbf{x}=(x_1,\ldots,x_n)\in \mathbb{F}_2^n$, the support of $\mathbf{x}$ will 
be denoted by $\sym{supp}(\mathbf{x})$ which is the set $\{i:x_i=1\}$; the weight of $\mathbf{x}$ will be denoted by $\sym{wt}(\mathbf{x})$ and is equal to 
$\#\sym{supp}(\mathbf{x})$.  For $i\in [n]$, $\mathbf{e}_i$ denotes the vector in $\mathbb{F}_2^n$ whose $i$-th component is 1 and all other components are 0. 
By $\mathbf{1}_n$ we will denote the all-one vector of length $n$.
%By $\mathbf{1}_n$ and $\mathbf{0}_n$ we will denote the all-one and all-zero vectors of length $n$ respectively. 
For $\mathbf{x}=(x_1,\ldots,x_n),\mathbf{y}=(y_1,\ldots,y_n)\in\mathbb{F}_2^n$, 
the inner product $\langle \mathbf{x},\mathbf{y}\rangle$ of $\mathbf{x}$ and $\mathbf{y}$ is defined to be $\langle \mathbf{x},\mathbf{y}\rangle =x_1y_1\oplus \cdots \oplus x_ny_n$.

For a positive integer $n$, an $n$-variable Boolean function $f$ is a map $f:\mathbb{F}_2^n\rightarrow \mathbb{F}_2$. Variables will be written in upper case and vector of variables
in bold upper case. For $\mathbf{X}=(X_1,\ldots,X_n)$, an $n$-variable Boolean function $f$ will be written as $f(\mathbf{X})$ to denote the dependence on the
variables $X_1,\ldots,X_n$. 

The support of a Boolean function $f$ will be denoted by $\sym{supp}(f)$ which is the set $\{\mathbf{x}:f(\mathbf{x})=1\}$; the weight of $f$ will be denoted by $\sym{wt}(f)$ and 
is equal to $\#\sym{supp}(f)$. 
%The expectation of $f$, denoted as $\mathbb{E}(f)$ (taken over a uniform random choice of $\mathbf{x}\in \mathbb{F}_2^n$), is equal to $\sym{wt}(f)/2^n$. 
The function $f$ is said to be balanced if $\sym{wt}(f)=2^{n-1}$.

The Fourier transform of a function $\psi:\mathbb{F}_2^n\rightarrow \mathbb{R}$ is the map $\widehat{\psi}:\mathbb{F}_2^n\rightarrow \mathbb{R}$, which is defined as follows. 
For $\bm{\alpha}\in\mathbb{F}_2^n$,
\begin{eqnarray}
	\widehat{\psi}(\bm{\alpha}) 
	& = & \frac{1}{2^n} \sum_{\mathbf{x}\in\mathbb{F}_2^n} \psi(\mathbf{x})(-1)^{\langle \mathbf{x},\bm{\alpha} \rangle}. \label{eqn-fourier}
\end{eqnarray}
The (normalised) Walsh transform of an $n$-variable Boolean function $f$ is the map $W_f:\mathbb{F}_2^n\rightarrow \mathbb{R}$, which is defined as follows.
For $\bm{\alpha}\in\mathbb{F}_2^n$,
\begin{eqnarray*}
	W_f(\bm{\alpha}) & = & \frac{1}{2^n} \sum_{\mathbf{x}\in\mathbb{F}_2^n} (-1)^{f(\mathbf{x}) \oplus \langle \mathbf{x},\bm{\alpha} \rangle} 
\end{eqnarray*}
In other words, the Walsh transform of $f$ is the Fourier transform of $(-1)^f$. 

From Parseval's theorem, it follows that $\sum_{\bm{\alpha}\in \mathbb{F}_2^n} W_f^2(\bm{\alpha}) = 1$.
So the values $\left\{W_f^2(\bm{\alpha})\right\}_{\bm{\alpha}\in\mathbb{F}_2^n}$ can be considered to be a probability distribution on $\mathbb{F}_2^n$, which assigns to 
$\bm{\alpha}\in\mathbb{F}_2^n$, the probability $W_f^2(\bm{\alpha})$. 
For an $n$-variable Boolean function $f$, its Fourier entropy $H(f)$ and min-entropy $H_{\infty}(f)$ are defined as follows. 
\begin{eqnarray}
	H(f) = \sum_{\substack{\bm{\alpha}\in \mathbb{F}_2^n\\W_f^2(\bm{\alpha})\neq 0}}W_f^2(\bm{\alpha}) \log \frac{1}{W_f^2(\bm{\alpha})}, & & 
	H_{\infty}(f) = \min_{\substack{\bm{\alpha}\in \mathbb{F}_2^n\\W_f^2(\bm{\alpha})\neq 0}}\log \frac{1}{W_f^2(\bm{\alpha})}. \label{eqn-H_infty}
	%H(f) & = & \sum_{\substack{\bm{\alpha}\in \mathbb{F}_2^n\\W_f^2(\bm{\alpha})\neq 0}}W_f^2(\bm{\alpha}) \log \frac{1}{W_f^2(\bm{\alpha})}, \label{eqn-H} \\
	%H_{\infty}(f) & = & \min_{\substack{\bm{\alpha}\in \mathbb{F}_2^n\\W_f^2(\bm{\alpha})\neq 0}}\log \frac{1}{W_f^2(\bm{\alpha})}. \label{eqn-H_infty}
\end{eqnarray}
For an $n$-variable Boolean function $f$, its influence $\sym{Inf}(f)$ is defined as follows.
\begin{eqnarray}\label{eqn-inf}
	\sym{Inf}(f) & = & \sum_{i=1}^n\Pr_{\mathbf{x}\in\mathbb{F}_2^n}[f(\mathbf{x}) \neq f(\mathbf{x}\oplus \mathbf{e}_i)]. %= %\frac{1}{2}\left(1-C_f(\mathbf{e}_i) \right).
\end{eqnarray}
The connection of influence to the Walsh transform is given by the following result~\cite{DBLP:conf/focs/KahnKL88} (see also Theorem~2.38 in~\cite{o2014analysis}).
\begin{eqnarray}\label{eqn-inf-WT}
	\sym{Inf}(f) & = & \sum_{\bm{\alpha} \in \mathbb{F}_2^n} \sym{wt}(\bm{\alpha}) W_f^2(\bm{\alpha}).
\end{eqnarray}

We next state the two conjectures connecting entropy and influence.

\paragraph{The Fourier entropy/influence (FEI) conjecture~\cite{friedgut1996every}.} There exists a universal constant $C$ such that for any integer $n\geq 1$ and
for any $n$-variable Boolean function $f$, $H(f)\leq C\cdot\sym{Inf}(f)$. 

\paragraph{The Fourier Min-entropy/influence (FMEI) conjecture~\cite{o2011fourier}.} There exists a universal constant $D$ such that for any integer $n\geq 1$ and 
for any $n$-variable Boolean function $f$, $H_{\infty}(f)\leq D\cdot\sym{Inf}(f)$. \\

%Since $H_{\infty}(f)\leq H(f)$, it follows that the FMEI is weaker than the FEI conjecture in the sense that settling the FEI conjecture will also settle the FMEI conjecture,
%but not vice versa. 
%Note that the minimum value of the constant in the FMEI conjecture is likely to be less than the minimum value of the constant in the FEI conjecture.

\paragraph{Composition.}
%For positive integers $n$ and $k$, an $(n,k)$ vectorial Boolean function (also called an S-box) is a map $\mathscr{G}:\mathbb{F}_2^n\rightarrow \mathbb{F}_2^k$. The function
For positive integers $n$ and $k$, an $(n,k)$ vectorial Boolean function is a map $\mathscr{G}:\mathbb{F}_2^n\rightarrow \mathbb{F}_2^k$. The function
$\mathscr{G}$ can be written as $\mathscr{G}(\mathbf{X})=(g_1(\mathbf{X}),\ldots,g_k(\mathbf{X}))$, where $g_1,\ldots,g_k$ are $n$-variable Boolean functions. 
Given a $k$-variable Boolean function $f$ and an $(n,k)$ vectorial Boolean function $\mathscr{G}$, their composition is the $n$-variable Boolean function
$(f\circ \mathscr{G})(\mathbf{X})=f(g_1(\mathbf{X}),\ldots,g_k(\mathbf{X}))$.
The Walsh transform of $f\circ \mathscr{G}$ is given by the following result.
\begin{theorem}\label{thm-GS}\cite{DBLP:journals/tit/GuptaS05a}
	Let $\mathscr{G}$ be an $(n,k)$ vectorial Boolean function and $f$ be a $k$-variable Boolean function. Then for any $\mathbf{u}\in\mathbb{F}_2^n$,
\begin{eqnarray}\label{general_composition_theorem}
W_{f\circ \mathscr{G}}(\mathbf{u}) &= \sum_{\mathbf{v}\in\mathbb{F}_2^k}W_f(\mathbf{v})W_{(l_{\mathbf{v}}\circ \mathscr{G})}(\mathbf{u}),
\end{eqnarray}
where $(l_{\mathbf{v}}\circ \mathscr{G})(\mathbf{X})=\langle \mathbf{v},\mathscr{G}(\mathbf{X})\rangle$.
\end{theorem}

Let $k$ and $l$ be positive integers and $n=kl$. For $\mathbf{x}\in \mathbb{F}_2^n$ and $1\leq i\leq k$, by $\mathbf{x}^{(i)}$ we denote the vector 
$(x_{(i-1)l+1},\ldots,x_{il})\in \mathbb{F}_2^l$. By a slight abuse of notation, we will write $\mathbf{x}=(\mathbf{x}^{(1)},\ldots,\mathbf{x}^{(k)})$.
Let $f$ and $g$ be Boolean functions on $k$ and $l$ variables respectively and $n=kl$. Let $\mathscr{G}$ be the $(n,k)$ vectorial Boolean function given by
$\mathscr{G}(\mathbf{X})=(g(\mathbf{X}^{(1)}),\ldots,g(\mathbf{X}^{(k)}))$.
The {\em disjoint composition} of $f$ and $g$, which we will denote as $f\diamond g$, is the $n$-variable Boolean function $f\circ \mathscr{G}$, i.e. 
\begin{eqnarray}
	(f\diamond g)(\mathbf{X}) & = & (f\circ \mathscr{G})(\mathbf{X})=f(g(\mathbf{X}^{(1)}),\ldots,g(\mathbf{X}^{(k)})).  \label{eqn-disj-comp}
\end{eqnarray}
The following result provides the entropy and influence of $f\diamond g$.
\begin{theorem}[simplified form of Proposition~$2$ in~\cite{o2013composition}]\label{thm-OT-comp}
	Let $f$ be a Boolean function and $g$ be a balanced Boolean function. Then,
	\begin{enumerate}
		\item $\sym{Inf}(f\diamond g) = \sym{Inf}(g)\cdot \sym{Inf}(f)$.
		\item $H(f\diamond g) = H(f) + H(g)\cdot \sym{Inf}(f)$.
	\end{enumerate}
 %\begin{eqnarray*}
 %H(f\diamond g) & = & H(f) + H(g)\cdot \sym{Inf}(f) \\
 %\sym{Inf}(f\diamond g) &=& \sym{Inf}(g)\cdot \sym{Inf}(f).
 %\end{eqnarray*}
\end{theorem}

\paragraph{O'Donnell-Tan recursive construction.}
The following recursive construction of Boolean functions was introduced by O'Donnell and Tan~\cite{o2013composition}. Let $g$ be an $l$-variable Boolean function. Using
$g$, a sequence of Boolean functions $f_m$, $m\geq 0$, is defined in the following manner.
\begin{eqnarray}\label{eqn-OT-rec}
	\left.\begin{array}{rcll}
	f_0 & = & g, \\
	f_m & = & g \diamond f_{m-1} & \mbox{if } m\geq 1.
	\end{array}\right\}
\end{eqnarray}
It is easy to see that for $m\geq 0$, $f_m$ is a map from $\mathbb{F}_2^{l^{m+1}}\rightarrow \mathbb{F}_2$. 
For the recursion defined in~\eqref{eqn-OT-rec}, in the case where the initial function $g$ is balanced, the following was proved in~\cite{o2013composition}.
\begin{eqnarray}\label{eqn-OT-ratio}
	\frac{H(f_m)}{\sym{Inf}(f_m)} & = & \frac{H(g)}{\sym{Inf}(g)} + \frac{H(g)}{\sym{Inf}(g)(\sym{Inf}(g)-1)} - \frac{H(g)}{\sym{Inf}(g)^{m+1}(\sym{Inf}(g)-1)}.
\end{eqnarray}
Consequently, $\lim_{m\rightarrow \infty} H(f_m)/\sym{Inf}(f_m) = H(g)/(\sym{Inf}(g)-1).$
So for any balanced Boolean function $g$, $H(g)/(\sym{Inf}(g)-1)$ is a lower bound on the constant in the FEI conjecture.
\begin{remark}\label{rem-bal-g-fm}
	It was shown in~\cite{o2013composition} that for the construction in~\eqref{eqn-OT-rec} if the initial function $g$ is balanced, then $f_m$ is balanced for all $m\geq 1$.
\end{remark}
%\begin{eqnarray}\label{eqn-OT-lim}
%	\displaystyle \lim_{m\rightarrow \infty} \frac{H(f_m)}{\sym{Inf}(f_m)} & = & \frac{H(g)}{\sym{Inf}(g)-1}.
%\end{eqnarray}
%A consequence of~\eqref{eqn-OT-lim} is that for any Boolean function $g$, $H(g)/(\sym{Inf}(g)-1)$ is a lower bound on the constant in the FEI conjecture.

\section{Min-Entropy of disjoint composition\label{sec-disj-comp-WT}}
We wish to compute the min-entropy of disjoint composition. We start with the following result which is somewhat more general than what we need.

\begin{theorem}\label{thm-WT-DC}
	Let $k$ and $l$ be positive integers and $n=kl$. Let $\mathscr{G}$ be an $(n,k)$ vectorial Boolean function such that 
	$\mathscr{G}(\mathbf{X})=(g_1(\mathbf{X}^{(1)}),\ldots,g_k(\mathbf{X}^{(k)}))$, where $g_1,\ldots,g_k$ are $l$-variable
	balanced Boolean functions. Then for any $k$-variable Boolean function $f$, 
\begin{eqnarray}\label{general_disjoint_composition_theorem}
	\begin{array}{l}
		W_{f\circ \mathscr{G}}(\mathbf{u})=
		\left\{ \begin{array}{ll}
				W_f(\mathbf{0}_k) &\text{ if $\mathbf{u}=\mathbf{0}_{n}$,}\\
				W_f(\mathbf{w_u})\prod_{i\in \sym{supp}({\mathbf{w_u}})}W_{g_i}\left(\mathbf{u}^{(i)}\right) &\text{ otherwise.}
			\end{array} \right. 
	\end{array}
\end{eqnarray}
In~\eqref{general_disjoint_composition_theorem}, 
for $\mathbf{u}\in \mathbb{F}_2^n$ written as $\mathbf{u}=(\mathbf{u}^{(1)},\ldots,\mathbf{u}^{(k)})$, by $\mathbf{w}_{\mathbf{u}}$ we 
denote the vector in $\mathbb{F}_2^k$ whose $i$-th position, $1\leq i\leq k$, is 1 if and only if $\mathbf{u}^{(i)}\neq \mathbf{0}_l$, i.e. $\mathbf{w}_{\mathbf{u}}$ encodes 
whether the $l$-bit blocks of $\mathbf{u}$ are zero or not.
%Consequently, $f\circ \mathscr{G}$ is balanced if and only if $f$ is balanced.
\end{theorem}

\begin{proof}
	The proof follows from an application of Theorem~\ref{thm-GS}. 

Note that for $\mathbf{v}=(v_1,\ldots,v_k)\in\mathbb{F}_2^k$, 
	$(l_{\mathbf{v}}\circ \mathscr{G})(\mathbf{X})=v_1\cdot g_1(\mathbf{X}^{(1)})\oplus\cdots \oplus v_k\cdot g_k(\mathbf{X}^{(k)})$.
So for $\mathbf{u}=(\mathbf{u}^{(1)},\ldots,\mathbf{u}^{(k)})\in\mathbb{F}_2^{n}$,
\begin{flalign*}
W_{(l_{\mathbf{v}}\circ \mathscr{G})}(\mathbf{u}) &= \frac{1}{2^{n}}\sum_{\mathbf{x}\in\mathbb{F}_2^{n}}(-1)^{(l_{\mathbf{v}}\circ \mathscr{G})(\mathbf{x})\oplus \langle\mathbf{u},\mathbf{x}\rangle} \\
%&=\frac{1}{2^{n}}\sum_{\mathbf{x}\in\mathbb{F}_2^{n}}(-1)^{(l_{\mathbf{v}}\circ \mathscr{G})(\mathbf{x})\oplus \langle\mathbf{u}^{(1)},\mathbf{x}^{(1)}\rangle\oplus\ldots\oplus\langle\mathbf{u}^{(k)},\mathbf{x}^{(k)}\rangle}\\
&= \frac{1}{2^{n}}\sum_{\mathbf{x}^{(1)},\ldots,\mathbf{x}^{(k)}\in\mathbb{F}_2^{l}}(-1)^{v_1\cdot g_1\left(\mathbf{x}^{(1)}\right)\oplus\langle\mathbf{u}^{(1)},\mathbf{x}^{(1)}\rangle\oplus\cdots\oplus v_k\cdot g_k\left(\mathbf{x}^{(k)}\right)\oplus\langle\mathbf{u}^{(k)},\mathbf{x}^{(k)}\rangle} \\
&= \prod_{i\in[k]}\frac{1}{2^{l}}\sum_{\mathbf{x}^{(i)}\in\mathbb{F}_2^{l}}(-1)^{v_i\cdot g_i\left(\mathbf{x}^{(i)}\right)\oplus\langle\mathbf{u}^{(i)},\mathbf{x}^{(i)}\rangle}.
\end{flalign*}
For $i\in[k]$, let
%\begin{equation*}
    $B_i\left(v_i,\mathbf{u}^{(i)}\right)
	=\frac{1}{2^{l}}\sum_{\mathbf{x}^{(i)}\in\mathbb{F}_2^{l}}(-1)^{v_i\cdot g_i\left(\mathbf{x}^{(i)}\right)\oplus\langle\mathbf{u}^{(i)},\mathbf{x}^{(i)}\rangle}.$
%\end{equation*}
Using~(\ref{general_composition_theorem}) we have,
\begin{equation}\label{inter_equation}
W_{f\circ \mathscr{G}}(\mathbf{u}) = \sum_{\mathbf{v}\in\mathbb{F}_2^k}W_f(\mathbf{v})W_{(l_{\mathbf{v}}\circ \mathscr{G})}(\mathbf{u})
= \sum_{\mathbf{v}\in\mathbb{F}_2^k}W_f(\mathbf{v})\prod_{i\in[k]}B_i\left(v_i, \mathbf{u}^{(i)}\right).
\end{equation}

Let us now consider $B_i\left(v_i,\mathbf{u}^{(i)}\right)$. Note that $B_i\left(0,\mathbf{u}^{(i)}\right)$ is equal to $1$ or $0$ according as
$\mathbf{u}^{(i)}$ is equal to $\mathbf{0}_l$ or not. Further, $B_i\left(1,\mathbf{u}^{(i)}\right)=W_{g_i}\left(\mathbf{u}^{(i)}\right)$. Since
it is given that $g_i$ is balanced, so $B_i\left(1,\mathbf{0}_l\right)=0$.

	For $\mathbf{u}\in \mathbb{F}_2^n$, the $i$-th bit of $\mathbf{w}_{\mathbf{u}}$ is 1 if and only if the $i$-th block of $\mathbf{u}$ is non-zero. 
	For $\mathbf{v}\in \mathbf{F}_2^k$ such that $\mathbf{v}\neq \mathbf{w}_{\mathbf{u}}$, there is a $j\in [k]$ such that either
	$v_j=0$ and $\mathbf{u}^{(j)}\neq \mathbf{0}_l$, or $v_j=1$ and $\mathbf{u}^{(j)}=\mathbf{0}_l$; in either case, $B_j\left(v_j,\mathbf{u}^{(j)}\right)=0$ and
	so $\prod_{i\in[k]}B_i\left(v_i, \mathbf{u}^{(i)}\right)=0$. On the other hand, for $\mathbf{v}=\mathbf{w}_{\mathbf{u}}$, if $v_i=0$ then 
	$\mathbf{u}^{(i)}= \mathbf{0}_l$ which implies $B_i\left(v_i,\mathbf{u}^{(i)}\right)=1$; and if $v_i=1$ then $\mathbf{u}^{(i)}\neq \mathbf{0}_l$ which implies 
	$B_i\left(v_i,\mathbf{u}^{(i)}\right)=W_{g_i}\left(\mathbf{u}^{(i)}\right)$; so
	$\prod_{i\in[k]}B_i\left(v_i, \mathbf{u}^{(i)}\right)=\prod_{i\in \sym{supp}(\mathbf{v})}W_{g_i}\left(\mathbf{u}^{(i)}\right)$.
From this, we get the required result.
\end{proof}
Suppose in Theorem~\ref{thm-WT-DC}, the $g_i$'s are all equal, i.e. $g_1=\cdots=g_k=g$. Then $f\circ \mathscr{G}=f\diamond g$ and Theorem~\ref{thm-WT-DC} provides
the Walsh transform of disjoint composition in the case where $g$ is balanced. %If further $f$ is also balanced, then $f\diamond g$ is balanced. 
In this case, the min-entropy is given by the following result.

\begin{theorem}\label{thm-ME-DC}
Let $k$ and $l$ be positive integers, $f$ be a $k$-variable Boolean function, and $g$ be an $l$-variable balanced Boolean function. 
	For $0\leq i\leq k$, let $a_i=\max_{\{\mathbf{w}:\sym{wt}(\mathbf{w})=i\}}W_f^2(\mathbf{w})$. Then
\begin{equation*}
	H_{\infty}(f\diamond g) = \min_{i\in \{0,\ldots,k\}, \\ a_i>0} (-\log(a_i) + i\cdot H_{\infty}(g)).
\end{equation*}
%where $\displaystyle a_i=\max_{\sym{wt}(\mathbf{w})=i}W_f^2(\mathbf{w})$ for $i=0,\ldots,k$.
%where $\displaystyle a_i=\max_{\substack{\mathbf{w}\in\mathbb{F}_2^k\setminus\{\mathbf{0}_{l}\}\\\sym{wt}(\mathbf{w})=i}}W_f^2(\mathbf{w})$.
\end{theorem}
\begin{proof}
	Let $n=kl$ and $\mathscr{G}$ be the $(n,k)$ vectorial Boolean function 
	$\mathscr{G}(\mathbf{X})=(g(\mathbf{X}^{(1)}),\ldots,g(\mathbf{X}^{(k)}))$. Then $f\diamond g=f\circ \mathscr{G}$ and we can apply Theorem~\ref{thm-WT-DC} to obtain
	the Walsh transform of $f\diamond g$. We have from Theorem~\ref{thm-WT-DC}, $W_{f\diamond g}(\mathbf{0}_n)=W_f(\mathbf{0}_k)$, and 
	for $\mathbf{0}_n\neq \mathbf{u}\in \mathbb{F}_2^n$, 
\begin{equation*}
    W_{f\diamond g}(\mathbf{u}) = W_f(\mathbf{w_u})\prod_{j\in \sym{supp}(\mathbf{w_u})}W_{g}\left(\mathbf{u}^{(j)}\right).
\end{equation*}
	From (\ref{eqn-H_infty}), to obtain the min-entropy of $f\diamond g$, it is required to obtain $\max_{\mathbf{u}\in \mathbb{F}_2^n}(W_{f\diamond g}(\mathbf{u}))^2$.
	Let $\bm{\alpha}_i=\argmax_{\sym{wt}(\mathbf{w})=i}W_f^2(\mathbf{w})$ for $i\in [k]$
	and let $\bm{\beta}=\argmax_{\mathbf{v}}W_{g}^2(\mathbf{v})$ (breaking ties arbitrarily in both cases). 
	Note that $a_i=W_f^2(\bm{\alpha}_i)$ and $H_{\infty}(g)=-\log W_g^2(\bm{\beta})$.
	For $\sym{wt}(\mathbf{w}_{\mathbf{u}})=i$, the maximum value of 
	$\prod_{j\in \sym{supp}(\mathbf{w_u})}W_{g}^2\left(\mathbf{u}^{(j)}\right)$ is $\left(W_g^2(\bm{\beta})\right)^i$.
	So $\max_{\mathbf{0}_n\neq \mathbf{u}\in\mathbb{F}_2^n}(W_{f\diamond g}(\mathbf{u}))^2$ is equal to $\max_{i\in [k]}W_f^2(\bm{\alpha}_i) \left(W_g^2(\bm{\beta})\right)^i=\max_{i\in [k]}a_i\left(W_g^2(\bm{\beta})\right)^i$.
	The result now follows by taking logarithms.
\end{proof}

\section{Recursive constructions \label{sec-lo-bnd}}
We wish to obtain a Boolean function $f$ such that $H_{\infty}(f)/\sym{Inf}(f)$ is as high as possible. One way to obtain $f$ is to perform an exhaustive search.
Since the number of $n$-variable Boolean functions is $2^{2^n}$, it is difficult to carry out the search for $n>5$. For $n=5$, we have performed an exhaustive
search. This resulted in 3840 5-variable Boolean functions for which the min-entropy/influence ratio is $16/7$. All the 3840 functions turned out to be unbalanced.
For the purpose of illustration, we provide one of the 3840 functions that were obtained. 
\begin{example}\label{ex-srch} Let $h$ be the following 5-variable Boolean function. 
\begin{eqnarray}\label{eqn-5-srch}
	h(X_5,X_4,X_3,X_2,X_1) & = & X_4X_3 \oplus X_5X_2 \oplus X_5X_4X_1 \oplus X_5X_4X_2 \oplus X_5X_4X_3.
\end{eqnarray}
	For $h$ defined in~\eqref{eqn-5-srch}, $H_{\infty}(h)=4$, $\sym{Inf}(h)=7/4$ and so $H_{\infty}(h)/\sym{Inf}(h)=16/7$. 
\end{example}
The question now is whether it is possible to obtain a function whose min-entropy/influence ratio is greater than 16/7?
In this section, we describe the approaches based on recursive constructions which did not provide such a function. In the next section, we describe a method
which yields a function whose min-entropy/influence ratio is greater than that of $h$.

\subsection{O'Donnell and Tan's Construction \label{subsec-OT-cons}}
We first consider the recursive construction of Boolean functions arising from the O'Donnell-Tan construction since this construction proved to be useful
for the entropy/influence ratio. 
Using Theorem~\ref{thm-ME-DC}, we obtain the following result on the min-entropy of the O'Donnell-Tan recursive construction where the initial function satisfies
the condition that there is a vector of weight 1 for which the corresponding Walsh transform value is the maximum. 

\begin{theorem}\label{thm-min-ent-OT}
	Let $g$ be an $l$-variable balanced Boolean function for which there is a $\bm{\beta}\in\mathbb{F}_2^l$ with $\sym{wt}(\bm{\beta})=1$ such that
	$W_g^2(\bm{\beta})=\max_{\mathbf{v}}W_{g}^2(\mathbf{v})$. 
	For $m\geq 0$, let $f_m$ be the Boolean function constructed using~\eqref{eqn-OT-rec} with $f_0=g$. Then for $m\geq 0$,
	\begin{eqnarray}\label{eqn-Hmin-OT}
		H_{\infty}(f_m) & = & (m+1)\cdot H_{\infty}(g).
	\end{eqnarray}
	Consequently, 
\begin{equation}\label{eqn-OT-FMEI}
    \frac{ H_{\infty}(f_m)}{\sym{Inf}(f_m)}=\left(\frac{H_{\infty}(g)}{\sym{Inf}(g)}\right)\left(\frac{m+1}{\sym{Inf}(g)^{m}}\right).
\end{equation}
\end{theorem}
\begin{proof}
	Note that $H_{\infty}(g)=-\log(W_g^2(\bm{\beta}))$. 
	As observed in Remark~\ref{rem-bal-g-fm}, from the fact that $g$ is balanced it follows that $f_m$ is balanced for all $m\geq 1$. (This can also be seen
	from Theorem~\ref{thm-WT-DC}.) 

	We prove~\eqref{eqn-Hmin-OT} by induction on $m$. For $m=0$, this follows trivially. Suppose~\eqref{eqn-Hmin-OT} holds for some $m\geq 0$. 
	From Theorem~\ref{thm-ME-DC} and the fact that $f_{m+1}$ is balanced, 
	\begin{eqnarray}\label{eqn-tmp0}
		H_{\infty}(f_{m+1}) & = & H_{\infty}(g\diamond f_m)= \min_{i\in [l], a_i>0} (-\log(a_i) + i\cdot H_{\infty}(f_m)),
	\end{eqnarray}
	where $a_i=\max_{\sym{wt}(\mathbf{w})=i}W_g^2(\mathbf{w})$ for $i=1,\ldots,l$. For any $i\in [l]$, we have
	\begin{eqnarray}\label{eqn-tmp1}
		-\log(a_i)+ i\cdot H_{\infty}(f_m)\geq -\log(W_g^2(\bm{\beta})) + H_{\infty}(f_m) & = & H_{\infty}(g)+H_{\infty}(f_m)
	\end{eqnarray}
	and since $\bm{\beta}$ has weight 1, equality is 
	attained for $i=1$. So using the induction hypothesis, 
	\begin{eqnarray}\label{eqn-tmp2}
		H_{\infty}(f_{m+1}) & = & \min_{i\in [l],a_i>0} (-\log(a_i) + i\cdot H_{\infty}(f_m))=H_{\infty}(g)+H_{\infty}(f_m)=(m+2)H_{\infty}(g).
	\end{eqnarray}
	The proof of~\eqref{eqn-OT-FMEI} follows from~\eqref{eqn-Hmin-OT} and Theorem~\ref{thm-OT-comp}.
\end{proof}

To use Theorem~\ref{thm-min-ent-OT} as an amplifier of min-entropy/influence ratio it is required to obtain $m\geq 1$ such that
$H_{\infty}(f_m)/\sym{Inf}(f_m) > H_{\infty}(g)/\sym{Inf}(g)$ which holds if and only if $\sym{Inf}(g) < (m+1)^{1/m}$. For $m=1$, this condition
becomes $\sym{Inf}(g)<2$ and for higher values of $m$, the upper bound on $\sym{Inf}(g)$ is smaller. Comparing~\eqref{eqn-OT-ratio} with~\eqref{eqn-OT-FMEI}, we see that unlike 
the case of the entropy/influence ratio, increasing $m$ does not necessarily lead to a higher value of the min-entropy/influence ratio. In particular, 
the nice asymptotic analyses~\cite{o2013composition,hod2017improved} which has been done for the entropy/influence ratio is not applicable to the min-entropy/influence ratio.

To apply Theorem~\ref{thm-min-ent-OT}, we need an appropriate initial function $g$. We performed an exhaustive search over all possible 5-variable Boolean 
functions which satisfy the conditions of Theorem~\ref{thm-min-ent-OT}. For $m=1$, we obtained 384 functions such that
taking $f_0$ to be any of these functions leads to a 25-variable Boolean function $f_1$ with $H_{\infty}(f_1)/\sym{Inf}(f_1)=512/225\approx 2.275556$. 
Let $\mathcal{F}_5$ denote the set of these 384 functions. We will use the elements of $\mathcal{F}_5$
in Section~\ref{subsec-30} to build a set of 30-variable functions having the presently highest known value of the min-entropy/influence ratio.
As an example, we provide one element of $\mathcal{F}_5$. 
\begin{example}\label{ex-OT}
	Let $g$ be the following 5-variable Boolean function.
\begin{eqnarray}\label{eqn-OT-g}
	\lefteqn{g(X_5,X_4,X_3,X_2,X_1)} \nonumber \\
	& = & X_3X_2X_1 \oplus X_4 \oplus X_4X_1 \oplus X_4X_2 \oplus X_4X_2X_1 \oplus X_4X_3X_1 \oplus X_4X_3X_2 \nonumber \\
	& & \quad \oplus X_5 \oplus X_5X_1 \oplus X_5X_2X_1 \oplus X_5X_3 \oplus X_5X_3X_1 \oplus X_5X_3X_2 \oplus X_5X_4 \nonumber \\
	& & \quad \oplus X_5X_4X_1 \oplus X_5X_4X_2 \oplus X_5X_4X_3.
\end{eqnarray}
	The function $g$ defined in~\eqref{eqn-OT-g} is in $\mathcal{F}_5$. For $g$, $H_{\infty}(g)=4$, $\sym{Inf}(g)=15/8$. Taking $f_0=g$ and $f_1=f_0\diamond f_0$, from 
	Theorem~\ref{thm-min-ent-OT} we have $H_{\infty}(f_1)/\sym{Inf}(f_1)=32/15 \times 2/(15/8) = 512/225$. 
\end{example}
We note the following points.
\begin{enumerate}
	\item The 25-variable function $f_1$ obtained using the above method is not useful. The 5-variable function $h$ given in~\eqref{eqn-5-srch} obtained 
		using exhaustive search has a higher value of the min-entropy/influence ratio.
	\item In our search over all 5-variable Boolean functions, considering $m>1$ did not provide a result better than that obtained for $m=1$.
	\item In Theorem~\ref{thm-min-ent-OT}, the condition $\sym{wt}(\bm{\beta})=1$ is required to obtain the expression for $f_m$ given by~\eqref{eqn-Hmin-OT}.
		If $\sym{wt}(\bm{\beta})>1$, then equality will not be attained in~\eqref{eqn-tmp1} and so in turn will also not be attained in~\eqref{eqn-tmp2}.
		So for $\sym{wt}(\bm{\beta})>1$, we will have $H_{\infty}(f_m) < (m+1)\cdot H_{\infty}(g).$ Consequently, 
	considering $\sym{wt}(\bm{\beta})>1$, does not seem to lead to a higher value of the min-entropy/influence ratio. In our search over all 5-variable
		functions, compared to $\sym{wt}(\bm{\beta})=1$, allowing $\sym{wt}(\bm{\beta})>1$ did not lead to a higher value of the min-entropy/influence ratio.
\end{enumerate}

\section{Construction from palindromic functions \label{sec-palin}}
An $n$-variable Boolean function $g$ can be represented by a bit string of length $2^n$ in the following manner: for $i\in \{0,\ldots,2^n-1\}$, the
$i$-th bit of the string is $g(\bm{\alpha})$, where $\bm{\alpha}$ is the $n$-bit binary representation of $i$. We will denote the bit string representing $g$ also by $g$. 
The reverse of the bit string representation of $g$ is $g^r$, and $g^r$ is given by $g^r(X_n,\ldots,X_1)=g(1\oplus X_n,\ldots,1\oplus X_1)$. The following
simple result relates the Walsh transforms of $g$ and $g^r$.
\begin{proposition}\label{prop-rev}
Let $g$ be an $n$-variable Boolean function and $g^r$ be another $n$-variable Boolean function defined as
	$g^r(X_n,\ldots,X_1)=g(1\oplus X_n,\ldots,1\oplus X_1)$. Then for $\bm{\alpha}\in\mathbb{F}_2^n$, 
	$W_{g^r}(\bm{\alpha}) = (-1)^{\sym{wt}(\bm{\alpha})} W_g(\bm{\alpha})$.
\end{proposition}
\begin{proof}
	\begin{eqnarray*}
		W_{g^r}(\bm{\alpha}) 
		& = & \sum_{\mathbf{x}\in\mathbb{F}_2^n}(-1)^{\langle\bm{\alpha},\mathbf{x}\rangle \oplus g(\mathbf{1}_n\oplus \mathbf{x})} 
		= \sum_{\mathbf{y}\in\mathbb{F}_2^n}(-1)^{\langle\bm{\alpha},\mathbf{1}_n\rangle\oplus \langle\bm{\alpha},\mathbf{y}\rangle \oplus g(\mathbf{y})} 
		= (-1)^{\sym{wt}(\bm{\alpha})} W_g(\bm{\alpha}).
	\end{eqnarray*}
\end{proof}

Given an $n$-variable Boolean function $g$, we may construct an $(n+1)$-variable Boolean $f$ function in the following manner. Concatenate the bit
string representing $g$ and $g^r$ to obtain a bit string of length $2^{n+1}$. This string represents the desired $(n+1)$-variable Boolean function $f$. The
bit string representing $f$ is a palindrome and we call $f$ to be a palindromic function. The following construction is a little more
general than the method just described. For $b\in \mathbb{F}_2$, let
\begin{eqnarray} \label{eqn-gen-palin}
	g_b(X_{n+1},X_n,\ldots,X_1) & = & (1\oplus X_{n+1})g(X_n,\ldots,X_1) \oplus X_{n+1}(b\oplus g(1\oplus X_n,\ldots,1\oplus X_1)).
\end{eqnarray}
If $b=0$, then $g_0$ is the concatenation of $g$ and $g^r$ as described above, and if $b=1$, then $g_1$ is the concatenation of $g$ and the complement of $g^r$. 
The following result shows the relation between the relevant properties of $g$ and $g_b$.
\begin{proposition}\label{prop-palin}
	Let $g$ be an $n$-variable Boolean function and $b\in \mathbb{F}_2$. Let $g_b$ be the $(n+1)$-variable Boolean function constructed from $g$ and $b$
	using~\eqref{eqn-gen-palin}. Then the following holds.
	\begin{enumerate}	
		\item For $\bm{\beta}\in \mathbb{F}_2^{n+1}$, where $\bm{\beta}=(a,\bm{\alpha})$, with $a\in\mathbb{F}_2$ and $\bm{\alpha}\in\mathbb{F}_2^n$,
	\begin{eqnarray}\label{eqn-palin-WT}
		W_{g_b}(\bm{\beta}) & = & \left(\frac{(1+(-1)^{b+\sym{wt}(\bm{\beta})})}{2}\right) W_g(\bm{\alpha}).
	\end{eqnarray}
	\item $H_{\infty}(g_b) = H_{\infty}(g)$. 
	\item $\sym{Inf}(g_b) = \sym{Inf}(g) + \epsilon_b(g)$, where 
		$\displaystyle \epsilon_b(g) = \sum_{\substack{\bm{\alpha}\in\mathbb{F}_2^n\\\sym{wt}(\bm{\alpha})\not\equiv b\bmod 2}} W_g^2(\bm{\alpha})$.
	\end{enumerate}
\end{proposition}
\begin{proof}
	By definition
	\begin{eqnarray}\label{eqn-W_f}	
		W_{g_b}(\bm{\beta}) & = & \frac{1}{2^{n+1}}\sum_{\mathbf{x}\in\mathbb{F}_2^{n+1}}(-1)^{g_b(\mathbf{x})\oplus \langle \bm{\beta},\mathbf{x}\rangle}.
	\end{eqnarray}	
	We simplify the exponent in the sum.
	\begin{eqnarray}
		\lefteqn{g_b(x_{n+1},x_n,\ldots,x_1) \oplus \langle(a,\bm{\alpha}), (x_{n+1},x_n,\ldots,x_1)\rangle} \nonumber \\
		& = & (1\oplus x_{n+1})g(x_n,\ldots,x_1) \oplus 
				x_{n+1}(b\oplus g(1\oplus x_n,\ldots,1\oplus x_1)) \oplus \langle(a,\bm{\alpha}), (x_{n+1},x_n,\ldots,x_1)\rangle \nonumber \\
		& = & 
		\left\{ \begin{array}{ll}
			g(x_n,\ldots,x_1) \oplus \langle \bm{\alpha}, (x_n,\ldots,x_1)\rangle & \mbox{if } x_{n+1}=0, \\
			b\oplus g(1\oplus x_n,\ldots,1\oplus x_1) \oplus a \oplus \langle \bm{\alpha}, (x_n,\ldots,x_1)\rangle & \mbox{if } x_{n+1}=1.
		\end{array} \right. \label{eqn-tmp}
	\end{eqnarray}
	Writing $\mathbf{x}=(x_{n+1},\mathbf{y})$, where $x_{n+1}\in\mathbb{F}_2$ and $\mathbf{y}\in \mathbb{F}_2^n$, we simplify~\eqref{eqn-W_f} using~\eqref{eqn-tmp} as follows.
	\begin{eqnarray*}
		W_{g_b}(\bm{\beta}) & = & \frac{1}{2^{n+1}} \left( \sum_{\mathbf{y}\in \mathbb{F}_2^n} (-1)^{g(\mathbf{y}) \oplus \langle \bm{\alpha},\mathbf{y} \rangle} 
		+ (-1)^{a\oplus b} \sum_{\mathbf{y}\in \mathbb{F}_2^n} (-1)^{g(\mathbf{1}_n\oplus \mathbf{y}) \oplus \langle \bm{\alpha},\mathbf{y} \rangle}
		\right) \\
		& = & \frac{1}{2} \left(W_g(\bm{\alpha}) + (-1)^{a\oplus b} W_{g^r}(\bm{\alpha}) \right) \\
		& = & \frac{1}{2} \left(W_g(\bm{\alpha}) + (-1)^{a\oplus b}(-1)^{\sym{wt}(\bm{\alpha})}W_g(\bm{\alpha}) \right) \quad \mbox{(using Proposition~\ref{prop-rev})} \\
		& = & \frac{1}{2} \left(W_g(\bm{\alpha}) + (-1)^{b}(-1)^{\sym{wt}(a,\bm{\alpha})}W_g(\bm{\alpha}) \right). % \quad \mbox{(using Proposition~\ref{prop-rev})} \\
	\end{eqnarray*}
	This proves the first point. The second point follows directly from the first. 
	
	For the third point, we use~\eqref{eqn-inf-WT} to compute the influence of $g_b$ from its Walsh transform.
	\begin{eqnarray*}
		\sym{Inf}(g_b)
		& = & \sum_{a\in\mathbb{F}_2,\bm{\alpha}\in\mathbb{F}_2^n} \sym{wt}(a,\bm{\alpha}) W_{g_b}^2(a,\bm{\alpha}) \\
		& = & \sum_{a\in\mathbb{F}_2,\bm{\alpha}\in\mathbb{F}_2^n} \sym{wt}(a,\bm{\alpha}) \left(\frac{(1+(-1)^{b+\sym{wt}(a,\bm{\alpha})})}{2}\right)^2 W_g^2(\bm{\alpha}) \\
		& = & \sum_{\bm{\alpha}\in\mathbb{F}_2^n} \sym{wt}(\bm{\alpha}) \left(\frac{(1+(-1)^{b+\sym{wt}(\bm{\alpha})})}{2}\right)^2 W_g^2(\bm{\alpha}) \\
		& & + \quad \quad \sum_{\bm{\alpha}\in\mathbb{F}_2^n} (1+\sym{wt}(\bm{\alpha})) \left(\frac{(1-(-1)^{b+\sym{wt}(\bm{\alpha})})}{2}\right)^2 W_g^2(\bm{\alpha}) \\
		& = & \sum_{\bm{\alpha}\in\mathbb{F}_2^n,\sym{wt}(\bm{\alpha})\equiv b\bmod 2} \sym{wt}(\bm{\alpha}) W_g^2(\bm{\alpha}) \\
		& & + \quad \quad \sum_{\bm{\alpha}\in\mathbb{F}_2^n,\sym{wt}(\bm{\alpha})\not\equiv b\bmod 2} (1+\sym{wt}(\bm{\alpha})) W_g^2(\bm{\alpha}) \\
		& = & \sym{Inf}(g) + \epsilon_b(g).
	\end{eqnarray*}
\end{proof}
We note the following two points.
\begin{enumerate}
	\item The Walsh transform of $g_b$ is banded, i.e. it is zero for all vectors of weights congruent to $1-b$ modulo two. 
	\item From Parseval's theorem it follows that $0\leq \epsilon_b(g) \leq 1$. Further, also from Parseval's theorem, we have
		$\epsilon_0(g)+\epsilon_1(g)=1$ and so it is possible to choose $b\in\{0,1\}$ such that $\epsilon_b(g)\geq 1/2$.
\end{enumerate}

We recall two well known classes of Boolean functions. See~\cite{carlet2021boolean} for an extensive discussion on the various properties of these classes. 
Let $f$ be an $n$-variable Boolean function.
\begin{itemize}
	\item $f$ is said to be $t$-resilient, $0\leq t<n$, if $W_f(\bm{\alpha})=0$ for all $\bm{\alpha}$ with $\sym{wt}(\bm{\alpha})\leq t$. 
		%If $t=0$, then $g$ is balanced.
	\item $f$ is said to be plateaued, if $W_f(\bm{\alpha})$ takes the values $0,\pm c$, for some $c$. 
\end{itemize}
From~\eqref{eqn-inf-WT}, it follows that if $f$ is $t$-resilient, then $\sym{Inf}(f)\geq t+1$.

Next we present the main result of the paper.
\begin{theorem} \label{thm-main}
	Let $g$ be a balanced $n$-variable Boolean function, $b\in \mathbb{F}_2$ and $g_b$ be constructed from $g$ and $b$ as in~\eqref{eqn-gen-palin}. Let
	$G_b=g_b\diamond g$. Then
	\begin{eqnarray}\label{eqn-G_b}
		\frac{H_{\infty}(G_b)}{\sym{Inf}(G_b)} & = & \frac{\min_{i\in\{0,\ldots,n+1\},a_i>0}(-\log(a_i)+iH_{\infty}(g))}{\sym{Inf}(g)(\sym{Inf}(g)+\epsilon_b(g))},
	\end{eqnarray}
	where $a_i=\max_{\{\mathbf{w}:\sym{wt}(\mathbf{w})=i\}} W_{g_b}^2(\mathbf{w})$, $i=0,\ldots,n+1$.

	Further, suppose that there is a $t\geq 0$ such that $t\equiv b\bmod 2$ and $g$ is a plateaued $t$-resilient function, which is not $(t+1)$-resilient. Then
	\begin{eqnarray}\label{eqn-G_b-sp}
		\frac{H_{\infty}(G_b)}{\sym{Inf}(G_b)} & = & \frac{H_{\infty}(g)}{\sym{Inf}(g)} \left(\frac{t+3}{\sym{Inf}(g)+\epsilon_b(g)}\right).
	\end{eqnarray}
\end{theorem}
\begin{proof}
	Proposition~\ref{prop-palin} provides the expression for $\sym{Inf}(g_b)$ and Theorem~\ref{thm-OT-comp} provides the expression for $\sym{Inf}(G_b)$.
	The expression for $H_{\infty}(G_b)$ is obtained from Theorem~\ref{thm-ME-DC}. This shows~\eqref{eqn-G_b}.

	Now suppose $g$ is a $t$-resilient plateaued function such that $t\equiv b\bmod 2$. Since $g$ is plateaued, from~\eqref{eqn-palin-WT}, it follows
	that $g_b$ is also plateaued and for $a_i>0$, $-\log(a_i)=H_{\infty}(g)$. From the conditions $g$ is $t$-resilient and $t\equiv b\bmod 2$, it follows
	that $g_b$ is $(t+1)$-resilient. To see this, suppose $\bm{\beta}\in\mathbb{F}_2^{n+1}$ with $\sym{wt}(\bm{\beta})\leq t+1$. If $\sym{wt}(\bm{\beta})=t+1$,
	then since $t\equiv b\bmod 2$, we have $1+(-1)^{b+\sym{wt}(\bm{\beta})}=0$ and so $W_{g_b}(\bm{\beta})=0$; on the other hand, if $\sym{wt}(\bm{\beta})<t+1$,
	then writing $\bm{\beta}=(a,\bm{\alpha})$ with $a\in\mathbb{F}_2$ and $\bm{\alpha}\in\mathbb{F}_2^n$, and using the fact that $g$ is $t$-resilient, 
	it follows that $\sym{wt}(\bm{\alpha})\leq t$ and so $W_g(\bm{\alpha})=0$ which implies that $W_{g_b}(\bm{\beta})=0$.
	Further, since $g$ is not $(t+1)$-resilient, it follows that $g_b$ is not $(t+2)$-resilient. 
	Since $g_b$ is $(t+1)$-resilient, but not $(t+2)$-resilient, it follows that the minimum value of $i$ such that $a_i>0$ is $t+2$.
	Now using the fact that for $a_i>0$, $-\log(a_i)=H_{\infty}(g)$,
	we have $\min_{i\in\{0,\ldots,n+1\},a_i>0}(-\log(a_i)+iH_{\infty}(g))\geq H_{\infty}(g)+(t+2)H_{\infty}(g)=(t+3)H_{\infty}(g)$. This shows~\eqref{eqn-G_b-sp}.
\end{proof}

\subsection{Construction of a 30-variable Boolean function \label{subsec-30}}

By construction, if $g$ is an $n$-variable Boolean function, then the function $G_b$ in Theorem~\ref{thm-main} is an $n(n+1)$-variable Boolean function.
To use Theorem~\ref{thm-main} as an amplifier of min-entropy/influence ratio, it is required to have $H_{\infty}(G_b)/\sym{Inf}(G_b) > H_{\infty}(g)/\sym{Inf}(g)$.
If $g$ is a plateaued $t$-resilient function, then the last condition holds if and only if $t+3\geq \sym{Inf}(g) + \epsilon_b(g)$. Note, however, that
$\sym{Inf}(g)\geq t+1$ and so the condition $t+3\geq \sym{Inf}(g) + \epsilon_b(g)$ offers only a limited scope for amplification of the min-entropy/influence
ratio.

If $g$ is balanced, but not 1-resilient, i.e. $t=0$, then the amplification factor in Theorem~\ref{thm-main} is $3/(\sym{Inf}(g) + \epsilon_b(g))$. 
We compare this condition with the amplification factor for $m=1$ arising from the O'Donnell-Tan construction. From Theorem~\ref{thm-min-ent-OT}, 
the amplification factor in the O'Donnell-Tan construction is $2/\sym{Inf}(g)$. So if we use the same $g$ in both Theorems~\ref{thm-min-ent-OT} and~\ref{thm-main},
then the amplification provided by Theorem~\ref{thm-main} is greater if and only if $\sym{Inf}(g) > 2\epsilon_b(g)$. We have computationally verified that the last 
condition holds for all $g\in \mathcal{F}_5$ (for the definition of $\mathcal{F}_5$ see the discussion before Example~\ref{ex-OT}).
So if we take any of the functions in $\mathcal{F}_5$ as the initial function and apply
Theorem~\ref{thm-main}, we will obtain a function whose min-entropy/influence ratio is greater than what can be obtained by starting with the same initial
function and using one step of the O'Donnell-Tan construction.

As a concrete example, we consider the 5-variable function $g$ given in Example~\ref{ex-OT}. Using this $g$ and taking $b=0$, from~\eqref{eqn-gen-palin}, we obtain
a 6-variable function $g_0$. The function $G_0=g_0\diamond g$ is a 30-variable function. From Theorem~\ref{thm-main}, we have
$H_{\infty}(G_0)/\sym{Inf}(G_0)=128/45\approx 2.8444$. Starting with any of the 384 functions in $\mathcal{F}_5$ and applying Theorem~\ref{thm-main}, we obtain
a corresponding 30-variable function for which the min-entropy/influence ratio is also 128/45. This gives us a set of 384 30-variable functions each of which 
has min-entropy/influence ratio to be 128/45. Note that 128/45 is greater than 16/7, which is the maximum min-entropy/influence ratio that is achieved by
any 5-variable function (see Example~\ref{ex-srch} and the discussion preceeding it). Presently, 128/45 is the highest known value of min-entropy/influence ratio
that has been achieved. Correspondingly, 128/45 is also the best known lower bound on the universal constant of the min-entropy/influence conjecture.

\section{Some further search results\label{sec-srch-res}}
A Boolean function $f$ is said to be symmetric if it is invariant under any permutation of its input. The number of $n$-variable symmetric Boolean functions
is $2^{n+1}$.
O'Donnell et al.~\cite{o2011fourier} established the FEI conjecture for symmetric Boolean functions which also settles the FMEI conjecture for this class of functions.
Their proof showed that the entropy/influence ratio of any symmetric Boolean function is at most 12.04. We used exhaustive search to find the actual value of the ratio
for symmetric functions on $n$ variables with $n\leq 16$. 

For $n\geq 2$, let $A_n(X_1,\ldots,X_n)=X_1\cdots X_n$ (in terms of Boolean algebra $A_n$ is the AND function). It is easy to show (see~\cite{hod2017improved}) that
$H(A_n)/\sym{Inf}(A_n)<4$. Our search for $n\leq 16$ showed that if $f$ is an $n$-variable symmetric Boolean function, then $H(f)/\sym{Inf}(f)\leq H(A_n)/\sym{Inf}(A_n)$. 
This suggests that the ratio 12.04 that was achieved in the proof of~\cite{o2011fourier} is perhaps not the minimum possible value of the entropy/influence ratio for 
symmetric functions.

A Boolean function $f$ is said to be bent~\cite{rothaus1976bent} if all the Walsh transform values of $f$ are equal. Such functions can exist only if $n$ is even.
If $f$ is bent, then $H(f)=H_{\infty}(f)=n$. Further, $\sym{Inf}(f)=n/2$ (see~\cite{DBLP:journals/ccds/BiswasS21}). So for a bent function $f$, 
$H(f)/\sym{Inf}(f)=H_{\infty}(f)/\sym{Inf}(f)=2$.
Symmetric functions can be bent and the class of symmetric bent functions have been characterised~\cite{Sa94,MS02}. Our search for $n\leq 16$ 
showed that if $n$ is even, then for any $n$-variable 
symmetric Boolean function $f$, $H_{\infty}(f)/\sym{Inf}(f)\leq 2$ and equality is achieved if and only if $f$ is bent; on the other hand, if $n$ is odd, then
for any $n$-variable symmetric Boolean function $f$, $H_{\infty}(f)/\sym{Inf}(f)<2$.

Based on our observations, we put forth the following conjecture.
\begin{conjecture}\label{conj:fei_sym_1}
	Let $f$ be an $n$-variable symmetric Boolean function. Then 
	\begin{enumerate}
		\item $H(f)/\sym{Inf}(f)\leq H(A_n)/\sym{Inf}(A_n)$ and equality is achieved if and only if $f$ equals $A_n$.
		\item If $n$ is even, then $H_{\infty}(f)/\sym{Inf}(f)\leq 2$ and equality is achieved if and only if $f$ is bent; if $n$ is odd, then
			$H_{\infty}(f)/\sym{Inf}(f)<2$
	\end{enumerate}
\end{conjecture}
A closed form expression for the Walsh transform of symmetric Boolean function in terms of binomial coefficients is known~\cite{DBLP:journals/aaecc/CastroMS18}.
We could not, however, find a way to use this expression to settle the above conjecture. 
We also tried to apply the techniques from~\cite{o2011fourier} used for showing that the FEI conjecture holds for symmetric Boolean functions to settle 
Conjecture~\ref{conj:fei_sym_1},
but were not successful. The main problem is that the various inequalities used in the proof of~\cite{o2011fourier} do not seem to be sufficiently sharp to establish the 
bounds stated in the above conjecture. As mentioned above, Conjecture~\ref{conj:fei_sym_1} has been verified for $1\leq n\leq 16$. It is possible to experimentally 
verify the conjecture for additional values of $n$, but this is unlikely to provide any insight into how to settle the conjecture.

A Boolean function is said to be rotation symmetric if it is invariant under a cyclic shift of its input. 
It is not known whether the FEI (or the FMEI) conjecture holds for rotation symmetric Boolean functions. 
See~\cite{DBLP:journals/ipl/StanicaM03} for the number of rotation symmetric Boolean functions on $n$ variables.
We could perform an exhaustive search on rotation symmetric Boolean functions for $n\leq 7$. For $n=6$ and $n=7$, the maximum values of $H(f)/\sym{Inf}(f)$ 
are 3.739764 and 3.804357 respectively; and the maximum values of $H_{\infty}(f)/\sym{Inf}(f)$ are 2.168978 and 2.227449 respectively, where the maximums are
over all $n$-variable rotation symmetric Boolean function. Compared to symmetric Boolean functions, we see that the maximum value of the entropy/influence ratio
remains below 4, but the maximum value of the min-entropy/influence ratio is greater than 2. Since we could not run the experiment for higher values of $n$, we are
unable to put forward any conjecture for rotation symmetric Boolean functions.

\section{Concluding remarks \label{sec-conclu}}
Our work has opened the interesting topic of obtaining lower bounds on the universal constant of the FMEI conjecture. We have provided one method of constructing 
Boolean functions which provides the presently best known lower bound. A future challenge is to obtain other construction methods which yield functions with a 
higher value of the min-entropy/influence ratio. 
It is also interesting to look for sufficiently sharp techniques to settle Conjecture~\ref{conj:fei_sym_1}. A final open problem resulting from our work
is to settle the FEI conjecture for rotation symmetric Boolean functions.

\section*{Acknowledgements} We thank Mridul Nandi for comments.

%\bibliographystyle{plainurl}
%\bibliography{fmei.bib}

\end{document}